\documentclass[conference]{IEEEtran}
\IEEEoverridecommandlockouts
\usepackage{cite}
\usepackage{amsmath,amssymb,amsfonts}
\usepackage{algorithmic}
\usepackage{amsthm}
\usepackage{graphicx}
\usepackage{textcomp}
\usepackage{xcolor}
\newtheorem{definition}{Definition}
\usepackage{mathtools}

\def\BibTeX{{\rm B\kern-.05em{\sc i\kern-.025em b}\kern-.08em
    T\kern-.1667em\lower.7ex\hbox{E}\kern-.125emX}}
\newtheorem{theorem}{Insight}
\begin{document}

\title{Analyzing Geospatial Distribution in Blockchains\\
}

\author{\IEEEauthorblockN{Shashank Motepalli,
Hans-Arno Jacobsen}
\IEEEauthorblockA{Department of Electrical and Computer Engineering,
University of Toronto\\
shashank.motepalli@mail.utoronto.ca,
jacobsen@eecg.toronto.edu}}

\maketitle

\begin{abstract}
Blockchains are decentralized; are they genuinely? We analyze blockchain decentralization's often-overlooked but quantifiable dimension: geospatial distribution of transaction processing. Blockchains bring with them the potential for geospatially distributed transaction processing. They enable validators from geospatially distant locations to partake in consensus protocols; we refer to them as minority validators. Based on our observations, in practice, most validators are often geographically concentrated in close proximity. Furthermore, we observed that minority validators tend not to meet the performance requirements, often misidentified as crash failures. Consequently, they are subject to punishment by jailing (removal from the validator set) and/or slashing (penalty in native tokens). Our emulations, under controlled conditions, demonstrate the same results, raising serious concerns about the potential for the geospatial centralization of validators. To address this, we developed a solution that easily integrates with consensus protocols, and we demonstrated its effectiveness.\end{abstract}


\section{Introduction}

A blockchain consists of a peer-to-peer network of validators (or miners) maintaining an immutable ledger. This immutable ledger contains a sequence of blocks linked to each other by appending the cryptographic hash of the last block to the subsequent block~\cite{nakamoto2008bitcoin}. Unlike centralized systems, blockchains foster \textit{decentralization}. Decentralization means processing transactions on a distributed ledger based on cryptographic proofs among a permissionless network of validators. This system eliminates the need for trusted centralized entities~\cite{nakamoto2008bitcoin, buterin2014next}.
Furthermore, decentralization is the fundamental force driving blockchain adoption~\cite{zarrin2021blockchain, chen2020blockchain}. However, despite the widespread usage of the term decentralization, it is unclear if blockchains have achieved the ideal state of decentralization, often considered the 'holy grail' of this technology. This lack of clarity stems from the complexity of quantifying decentralization due to its multi-disciplinary nature, requiring an understanding of technology, society, economics, and politics. This work focuses on a vital but often overlooked dimension contributing to decentralization: the geospatial distribution of transaction processing. 

The geospatial diversity of validators, who process transactions, is vital for blockchains and decentralization. Firstly, boosting geospatial diversity contributes to the \textit{robustness} of the blockchain. In other words, the blockchain would be more resilient to downtimes caused by various factors, such as the lack of electricity required for computation, changing regulations, and even wars. Imagine most validators of a blockchain network located at a data center that does not adhere to fire safety or is prone to natural calamities. This situation is not far from reality; for instance, in September of 2021, an outage at an AWS data center could bring down the entire Solana blockchain~\cite{amazonSolana}. Even worse, a cloud provider chooses to bring down 40\% of all validators of a live blockchain for not adhering to its policies~\cite{HetznerSolana}. Secondly, geospatial diversity facilitates \textit{fairness} of the blockchain resources, notably for time-critical operations, by lowering latencies and providing equitable access to users from all geographic regions. For illustration, consider the advantages of arbitrage traders on a decentralized exchange, such as UniSwap~\cite{adams2021uniswap}, due to geographic proximity to validators. Thirdly, promoting geospatial diversity can reduce the burden of complying with the regulations of jurisdictions to which one does not belong. For instance,  according to SEC, everyone who does transactions on Ethereum falls under the US jurisdiction because the US has around half of the network's validators~\cite{SECEthereum}. This practice is worrisome as one might break the law without being aware of the blockchains' underlying infrastructure. Existing blockchains fail to deliver geospatial diversity. For instance, in the case of Ethereum, a prominent blockchain, 69.91\% of total validators are in two regions, namely the US and Germany~\cite{etherscan}; see Fig.~\ref{fig:ethereumDistribution}.
\begin{figure}[]
  \centering
  \includegraphics[width=0.5\textwidth]{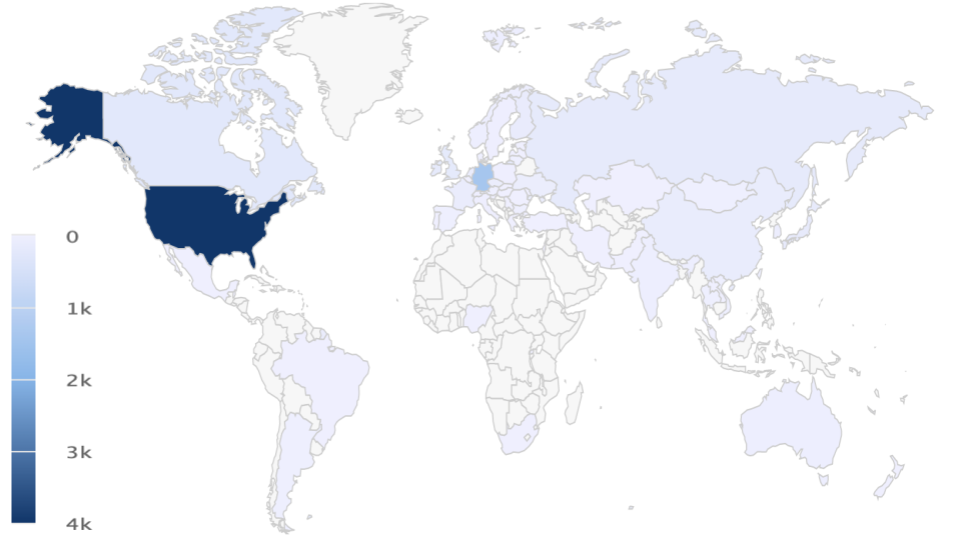}
  \caption{Distribution of validators in the Ethereum network. 53.11\% in the US and 18.29\% in Germany. Source: Etherscan~\cite{etherscan}.}
  \label{fig:ethereumDistribution}
\end{figure}

This paper addresses the problem of \textit{analyzing the geospatial distribution of blockchains} and the factors restraining geospatial diversity in a blockchain network. This problem is relevant to proof of stake (PoS) blockchains that rely on classical consensus mechanisms for creating new blocks, such as Ethereum~\cite{buterin2017casper}, Cosmos~\cite{buchman2016tendermint}, Diem~\cite{baudet2019state}. What makes the problem interesting is that the geospatial location of validators is usually overlooked but is a measurable attribute for decentralization. The interplay between the underlying consensus protocols and networking among validators makes it a challenging problem to dissect. Nevertheless, unlocking new levels of robustness and fairness for blockchains makes it worth pursuing. 

This paper presents GeoDec, an emulator for blockchains' underlying consensus protocol to measure validators' performance and their rewards and penalties. GeoDec helps anyone understand the interplay between a validator's location and performance in a given environment. GeoDec takes in the geospatial locations of validators as parameters and introduces latencies among them based on their ping delays as if the validators' machines were physically present at the given locations. 
Furthermore, by running GeoDec emulations, we can show and empirically prove that validators geospatially distant from the other validators are disadvantaged. Finally, we propose and evaluate solutions to promote geospatial diversity that integrates with consensus protocols.

The contributions of this paper are four-fold. (1) Our key contribution is providing empirical evidence that minority validators are disadvantaged by design, both in practical blockchain applications and in emulations. (2) 
To our knowledge, we are among the first to build a blockchain consensus protocol emulator for analyzing geospatial distribution. (3) We provide an integrative solution to prevent geospatially distant validators from being punished. (4) We also introduce standardized metrics for measuring geospatial diversity for blockchains.

This paper is organized as follows. Section~\ref{sec:Background} provides the necessary background about consensus protocols and validator performance. We then provide empirical evidence from blockchains in practice in Section~\ref{sec:0LMotivation}. In Section~\ref{sec:emulatorGDI}, we design the GeoDec emulator and the GDI index. We show that minority validators are punished, then propose and evaluate our solutions in Section ~\ref{sec:saveminority}. In Section~\ref{sec:relatedwork}, we review related work. Finally, in Section~\ref{sec:conclusions}, we provide conclusions and discuss potential future directions for this research.

\section{Background}
\label{sec:Background}
\subsection{Blockchain and Consensus Protocols}
A blockchain is a sequence of blocks, and the process of reaching an agreement on the order of this sequence is known as the \textit{consensus protocol}. Validators, denoted by \begin{math}\textit{V} = {v_1, v_2, ..., v_t}\end{math}, are participants in the consensus protocol. They are responsible for proposing these blocks and verifying the correctness of each block. This set of validators is referred to as the validator set.

An epoch, denoted by $E_\tau$, is a fixed time interval $\delta$ during which the validator set \textit{$V_{E_\tau}$} remains consistent. After this interval, we move to a new epoch $E_{\tau+1}$ with a potentially different validator set \textit{$V_{E_{\tau+1}}$}.

Consensus protocols are at the heart of blockchains, and they must be Byzantine fault-tolerant (BFT) to withstand malicious validators in the validator set. These protocols can be broadly divided into two categories based on how they achieve finality. Probabilistic consensus mechanisms, such as Nakamoto consensus, guarantee finality over time, with the probability of finality increasing as more successive blocks are added. On the other hand, classical consensus mechanisms achieve instant definitive finality, even if up to 1/3 of the validator set is faulty. They accomplish this by requiring a supermajority quorum on every block, 
\begin{math} |\mathbb{Q}| \geq \lceil{(2/3)*|\textit{V}|} \end{math}.

For the purposes of this paper, we focus on blockchains that rely on classical consensus protocols. Since the formulation of the Byzantine generals problem in the 1980s~\cite{lamport2019byzantine}, we have seen numerous classical consensus protocols offering high throughput, low latency, and robust system designs~\cite{zhang2022reaching}. Among all classical consensus protocols, we choose HotStuff, a leader-based BFT consensus protocol for partial synchrony, for multiple reasons: (1) Linear communication complexity, reaching consensus in a single round in the best case with frequent leader rotation and threshold signature aggregation~\cite{yin2019hotstuff}. (2) Offers optimistic responsiveness; we reach consensus on a block once we have a quorum $\mathbb{Q}$ with no further delays~\cite{yin2019hotstuff}. (3) Experimental evaluations~\cite{alqahtani2021bottlenecks} proved that HotStuff offers higher throughput and lower latency than Tendemint~\cite{buchman2016tendermint} and PBFT~\cite{castro1999practical}. (4) Practical implementation on blockchains such as Diem, Aptos, and 0L~\cite{baudet2019state}.

In HotStuff, every block has a block proposer responsible for proposing the block and aggregating votes to build a quorum $\mathbb{Q}$. This protocol uses point-to-point communication between validators, requiring a complete mesh network topology with $|V|*(|V|-1)$ networking connections among validators.

By using PoS as a Sybil resistance mechanism, classical consensus can be applied in a decentralized environment~\cite {buchman2016tendermint}. However, in the context of this paper, we assume all validators have an equal stake, i.e., equal voting power in reaching the quorum. With minor modifications, this work can be extended to situations where validators have unequal weights in consensus.

\subsection{Validator Participation in Consensus}
In classical consensus protocols, each block requires the achievement of a quorum, denoted as $\mathbb{Q}$. It is reasonable to anticipate that not all validators will adhere to the protocol due to either malicious behavior or crash failures. This paper focuses on the latter, where a validator fails to actively participate in the consensus mechanism due to crash failures. Such validators can induce delays and, in some cases, impede the liveness of consensus protocols by halting block production~\cite{cohen2022aware, abraham2022s}. To maintain the integrity of the blockchain, these behaviors need to be penalized~\cite{motepalli2021reward}. Most blockchains, such as Ethereum~\cite{eth2stake}, Polkadot~\cite{polkadot}, and Cosmos~\cite{cosmos}, have implemented penalization mechanisms. These penalties can take the form of a fine paid in native tokens (slashing) or temporary removal from the validator set \textit{V} for a certain number of epochs (jailing).

We use the term "liveliness" to measure validator participation in the consensus protocol. \textit{Liveliness} is defined as the ratio of the number of blocks signed by a validator to the total number of blocks committed within an epoch~\cite{motepalli2022decentralizing}. If B($E_\tau$) represents the total number of committed blocks in an epoch $E_\tau$ and B($E_\tau$, $v_i$) represents the number of blocks signed by the validator $v_i$, the liveliness $l$ of validator $v_i$ in the epoch $E_\tau$ can be calculated as follows:
\begin{equation}
l_{E_\tau}(v_i) = \dfrac{B(E_\tau, v_i)}{B(E_\tau)} * 100
\end{equation}

This metric provides an indication of a validator's activity during a given epoch. At the end of each epoch, any validator that fails to meet a predefined liveliness threshold $\pi$ is considered to have experienced a crash failure and is subject to a penalty. It's important to note that we consider the number of committed blocks in this calculation rather than block proposals, as some proposed blocks may time out and never be committed.

\section{Empirical Observations From 0L}
\label{sec:0LMotivation}
The data underpinning this section is derived from empirical observations of the 0L blockchain\footnote{https://0l.network/}, a permissionless fork of the Diem blockchain~\cite{baudet2019state, biel2021zero, motepalli2022decentralizing}, which employs the HotStuff consensus mechanism. The duration of an epoch in this blockchain network is set to a day, with $\delta$ = 24 hours. 0L excludes validators that fail to meet the liveliness threshold ($\pi = 5$) from the validator set for the subsequent epoch. Specifically, if a validator $v_i$ has a liveliness $l$ that falls below $\pi$ at epoch $E_\tau$, 0L eliminates $v_i$ from the validator set for the following epoch $E_{\tau+1}$.

\begin{figure}[]
  \centering
  \includegraphics[width=0.5\textwidth, trim={0 0 7cm 0},clip]{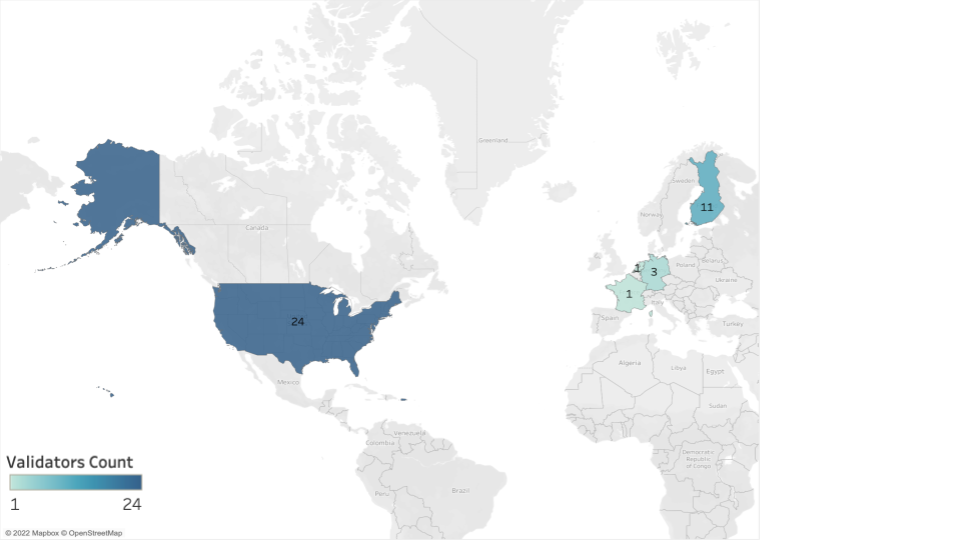}
  \caption{Distribution of validators of 0L, snapshot taken at epoch 317.}
  \label{fig:0Ldistribution_epoch317}
\end{figure}

Figure ~\ref{fig:0Ldistribution_epoch317} illustrates the geospatial distribution of validators at epoch 317. Out of 39 validators, most are located in the US (58.9\%) and Finland (28.2\%), with the remainder in Germany, France, and the Netherlands. The geospatial distribution of the 0L blockchain was monitored over the next 31 epochs, and the distribution was found to remain stable. At the end of epoch 348, the validator set size was 43, with 51.1\%, 30.2\%, and 11.6\% of validators located in the US, Finland, and Germany, respectively. There are single validators in Canada, France, and Lithuania. The central question is why all validators are concentrated in a few geospatial locations, and whether the underlying consensus protocol plays a role in this.

Further analysis revealed that validators from Singapore and Australia, despite joining the network, could not maintain their position in the validator set as they failed to meet the liveliness threshold ($l < \pi$) for the epochs during which they were part of the validator set $V$. A majority of validators voluntarily submitted their logs via Prometheus and Grafana, this data indicated that the validators from Singapore had not experienced crash-faults; they had synchronized the latest state of the ledger and maintained both outbound and inbound connections with most validators throughout the epoch. These findings led to the hypothesis that the consensus protocol is indeed responsible for limiting geospatial diversity.

The hypothesis posits that\textit{ due to the higher latency of geospatially distant validators, quora $\mathbb{Q}$ tend to form among localized validator communities. As a result, votes from geospatially distant validators are not counted towards the quorum}, leading to lower liveliness and perceived crash failure. This failure to meet the liveliness threshold ($l < \pi$) results in penalties for geospatially distant validators. In this context, the consensus protocol inadvertently undermines decentralization, leading to geospatial centralization by design.

This concern is not specific to the 0L blockchain alone. Aptos, another permissionless blockchain based on HotStuff, had rewards based on liveliness. Their lessons from running test networks were that rewards based on liveliness disincentivized geospatial distribution~\cite{AptosTestnet}. Furthermore, this is a concern for classical consensus protocols other than HotStuff. For instance, Narwhal and Tusk, the latest classical consensus protocol offering high throughput, noted that slow validators are indistinguishable from faulty ones~\cite{danezis2022narwhal}. Furthermore, the authors note that this is a limitation of asynchronous protocols with optimistic responsiveness~\cite{danezis2022narwhal}. Besides, previous studies of classical consensus protocols, including PBFT~\cite{castro1999practical}, focused on throughput and latencies and, hence, performed evaluations in local area networks. This approach overlooks how the protocol might perform in a wide area network. However, the liveliness of individual validators in a wide area network is vital for blockchains as it directly affects incentive structures for validators.

\section{GeoDec to Analyze Geospatial Distribution of Validators}
\label{sec:emulatorGDI}

\subsection{GeoDec Emulator}
\label{sec:geodecEmulator}
We introduce GeoDec, an emulator for studying the geospatial distribution of validators in a blockchain. GeoDec executes the underlying classical consensus mechanism, and we have chosen the HotStuff consensus protocol for its implementation~\cite{yin2019hotstuff}. GeoDec operates under the assumption that all validators are active, ruling out any crash or Byzantine failures.

GeoDec accepts a hashmap $\mathbb{C}$, representing the distribution of validators, as input. The keys in this hashmap are locations, expressed as cities, with the corresponding values being the number of validators in each location. We represent this distribution as $\mathbb{C} = { c_1: i, c_2: j, ..., c_m: p}$, where $c_1$, $c_2$, and $c_m$ represent cities and $i$, $j$, and $p$ are natural numbers denoting the number of validators at each location. Therefore, $\mathbb{C}[c_k]$ gives the number of validators in city $c_k$. The output of GeoDec is the liveliness $l$ of the validators, a measure of each validator's contribution to the consensus protocol. The liveliness $l$ depends not only on the location of a validator but also on the distribution of the rest of the validators.

The subsequent part of this section delves into the specifics of the GeoDec implementation. Developed in Python and using shell scripts to manage the validator cluster, GeoDec's code is open-sourced under Apache License 2.0 and hosted on GitHub\footnote{https://github.com/sm86/geodec-hotstuff}. The emulator runs on ComputeCanada~\cite{baldwin2012compute}, where a cluster of nodes acts as validators executing the HotStuff protocol~\cite{yin2019hotstuff}. The size of the validator set $|V|$ is adjustable, with a minimum of four validators, and an upper limit determined by the available computing resources. Nevertheless, the protocol's performance appears to deteriorate with more than a hundred validators. To ensure consistency, we use the same configurations across all nodes: 7.5GB RAM, two vCPUs, 56GB total disk, running on Ubuntu 20.04.5 LTS.

Several configurable parameters are available for benchmarking the underlying HotStuff protocol. We set the transaction size to 512 bytes and the batch size to 4000 transactions, as these parameters most closely resemble real-world blockchain infrastructures. With GeoDec assuming no crash failures or Byzantine faults, we increased timeouts to fifty seconds to ensure infinite time before a consensus round on any block times out. We retained the block proposer selection algorithm from HotStuff, choosing each block's proposer randomly using a round-robin mechanism. For repeatability, each GeoDec run outputs the mean of liveliness $l$ over five epochs, $E_\tau$ to $E_{\tau+5}$, with the duration of an epoch set to five minutes ($\delta$ = 300 seconds). The liveliness of a validator $l(v_k)$ for a run is given as:
\begin{equation}
l(v_k) = \frac{\sum_{t=\tau}^{\tau+5}l_{E_t}(v_k)}{5}
\end{equation}

GeoDec employs netem\footnote{https://wiki.linuxfoundation.org/networking/netem} to emulate network latencies between all pairs of validators. Although our implementation situates the validators within a local area network, we seek to emulate geographical distribution as if each validator were located at their designated physical location. Netem offers functionality for protocol testing by simulating the properties of a wide area network~\cite{hemminger2005network}. We use pairwise ping delays among 242 cities, collected by WonderProxy~\cite{wonderproxy}, to emulate a wide area network. During preprocessing, we grouped data by averaging the ping latencies recorded every hour for two days to determine pairwise ping delays. We also excluded a few cities lacking pairwise delays. GeoDec then uses this preprocessed data as input for netem to establish the pairwise delays between validators using IP addresses.

In conclusion, GeoDec is an emulator for blockchains operating on classical consensus protocols with preset parameters. One potential limitation of GeoDec is that it does not consider if some data centers in a given city have better connectivity than the rest within the same city. Nonetheless, GeoDec offers easily configurable parameters, making it adaptable for a range of applications. Moreover, it is infrastructure-agnostic, meaning it can be migrated to operate on different computing platforms such as AWS or Google Cloud.

\subsection{GeoSpatial Diversity Index}
\label{sec:gdi}
We propose the geospatial diversity index (GDI) as a metric to measure geospatial diversity. Why do we need such a metric? This metric is necessary to standardize the evaluation and analysis of geospatial distribution. Furthermore, GDI helps a blockchain measure its geospatial diversity and foster a dialogue around achieving more significant geospatial representation. 

Several characteristics are essential when considering the GDI metric. First, GDI should incorporate geospatial distance to capture validators' geospatial diversity. Second, each validator has a unique frame of reference and thus has a respective GDI. Third, a validator's GDI should depend not only on their performance but also on the location of the remaining validators. Finally, Earth is spherical, so the GDI should consider the distances based on latitudes and longitudes.

\begin{definition}
Geospatial Diversity Index of a Validator, $GDI(v_k)$: The GDI of a validator quantifies the geospatial diversity of validator $v_k$ relative to the rest of the validator set $V$. Specifically, $GDI(v_k)$ is the sum of the distances, from $v_k$, to all other validators within the set.
\end{definition}

We extract all validators' coordinates from the validator set's geospatial data, $\mathbb{C}$. We denote $\Delta(v_i, v_j)$ as the Haversine distance between validators $v_i$ and $v_j$. The Haversine distance is appropriate as it uses latitudes and longitudes to calculate the shortest path between two points along Earth's surface~\cite{robusto1957cosine}. The GDI for validator $v_k$ in epoch $E_\tau$ is then given as:

\begin{equation}
GDI_{E_\tau}^V(v_k) = \sum_{\forall v_i \in V_{E_\tau}}\Delta(v_i, v_k)
\end{equation}

$GDI_{V_{E_\tau}}(v_k)$ captures how far the validator $v_k$ is from the rest of the validator set. Units are in kilometers. As the validator set may change every epoch, the GDI for a validator may not be constant across epochs.

However, this equation has a significant limitation: it does not account for using a quorum in the consensus protocol. If a validator proposes a block, we only need two-thirds of the validator set $\mathbb{Q}$ to respond, not all validators. Therefore, we introduce $GDI_{\mathbb{Q}{E\tau}}(v_k)$ to capture the geospatial distance from the closest quorum of validators. We define $\bar{V}$ to capture that only the nearest validators are needed to form the quorum. The validator set $V$ is sorted based on distance $\Delta$ and trimmed to two-thirds to obtain $\bar{V}$. In this work, we refer to $GDI^{\mathbb{Q}}(v_k)$ whenever we mention $GDI(v_k)$.

\begin{equation}
GDI_{E_\tau}(v_k) = GDI_{E_\tau}^{\mathbb{Q}}(v_k) = \sum_{\forall v_i \in \Bar{V}_{E_\tau}}\Delta(v_i, v_k)
\end{equation}

We also deem it crucial to ask if we can measure GDI for the entire blockchain, as such a measure would offer a holistic view of the entire validator set's geospatial diversity. We define GDI for a blockchain as follows. 
\begin{definition}
 Geospatial Diversity Index of a Blockchain. ($GDI$). The GDI of a blockchain is a measure of how geospatially diverse a validator set is. We calculate it by taking the mean of the GDI of every validator in the validator set $V$.
\end{definition}
GDI for a blockchain is defined for an epoch, because the validator set might change every epoch. GDI for an epoch $E_\tau$ with validator set $V_{E_\tau}$ is given below. 
\begin{equation}
GDI_{E_\tau} = \frac{\sum_{\forall v_k \in V_{E_\tau}}GDI_{E_\tau}(v_k)}{|V_{E_\tau}|}
\end{equation}

So far, this section has focused on GeoDec and the GDI measures. We now leverage these to identify geospatial outliers; we start by defining minority validators. 

\begin{definition}
  Minority ($\mu$) is the location of the most geospatially distant validators. The minority set is subset of all cities $\mu \subset \mathbb{C}$.
\end{definition}

We refer to validators in minority cities as minority validators. We use the GDI of validators to identify minorities in the given validator set. A minority is a set of locations with the validators' GDI being higher than the $67^{th}$ percentile. We use percentiles to capture an individual GDI relative to the rest of the validator set. We choose the $67^{th}$ percentile because the cardinality of the minority validator set has to be less than one-third of the validator set, $\mathbb{C}[\mu] \leq (1/3 * |V|)$. To obtain the $67^{th}$ percentile, we order the GDI values of all validators in $V$ and mark the highest one-third as the minority $\mu$. Similarly, we now define the majority. 

\begin{definition}
  Majority ($\eta$). The majority refers to the location in the validator set that has enough validators in its reach to form a quorum $Q$. In other words, if over two-thirds of the validator set is in one city $c_p \in C$, i.e., $C[c_p] \geq (2/3 * |V|)$, we call that city the majority $\eta = c_p$. 
\end{definition}

We have, at most, one majority for any given validator set. The presence of a majority implies the geospatial centralization in that location. The output from a GeoDec execution comprises GDI measures and the liveliness.      

\section{Save the Minority Validators}
\label{sec:saveminority}
This section builds on the introduction of the GeoDec emulator and metrics for understanding the geospatial distribution of a blockchain, discussed in the previous section. It further explores how the consensus protocol design impacts minority validators and presents potential solutions to this challenge.

\subsection{Challenge: Minorty validators are punished}
Having developed the GeoDec emulator, we now revisit our hypothesis on whether consensus protocols consider geographically distant validators as having crash failures (see Section \ref{sec:0LMotivation}). We emulate the validator distribution of 0L on GeoDec to see if we could reproduce similar results. 

\begin{figure}[]
  \centering
  \includegraphics[width=0.5\textwidth]{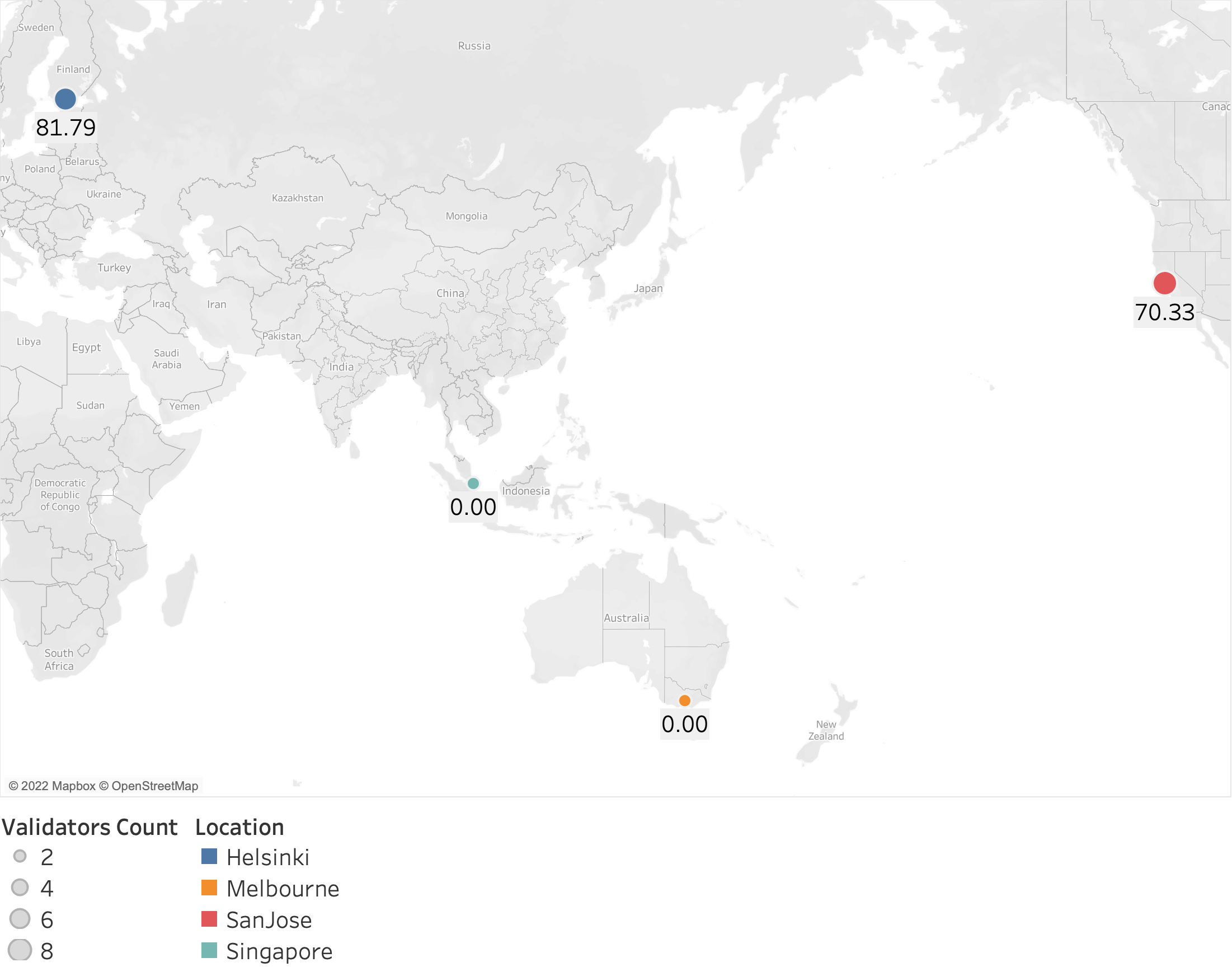}
  \caption{Liveliness data for where the most validators are in the US and Finland, and minority validators in Singapore and Australia.}
  \label{fig:twoMajorities}
\end{figure}

For our emulations, we fixed the validator set size to 16, concentrating most validators in two regions: San Jose, US, and Helsinki, Finland. We chose these locations as they tend to have the largest percentage of validators in 0L. In our first series of evaluations, we relocated a validator to Singapore resulting in the validator distribution $\mathbb{C} = \{ San Jose: 8,\:Helsinki:7,\:Singapore:1\}$. Subsequently, we ran evaluations with two validators in Singapore as minorities, resulting in the validator distribution, $\mathbb{C} = \{ San Jose: 7,\:Helsinki: 7,\; Singapore: 2\}$. Similarly, we conducted emulations for Melbourne, Australia, with $\mathbb{C}[Melbourne] = 1$ and $\mathbb{C}[Melbourne] = 2$, respectively.

\begin{figure*}[]
  \centering
  \includegraphics[ scale=0.22]{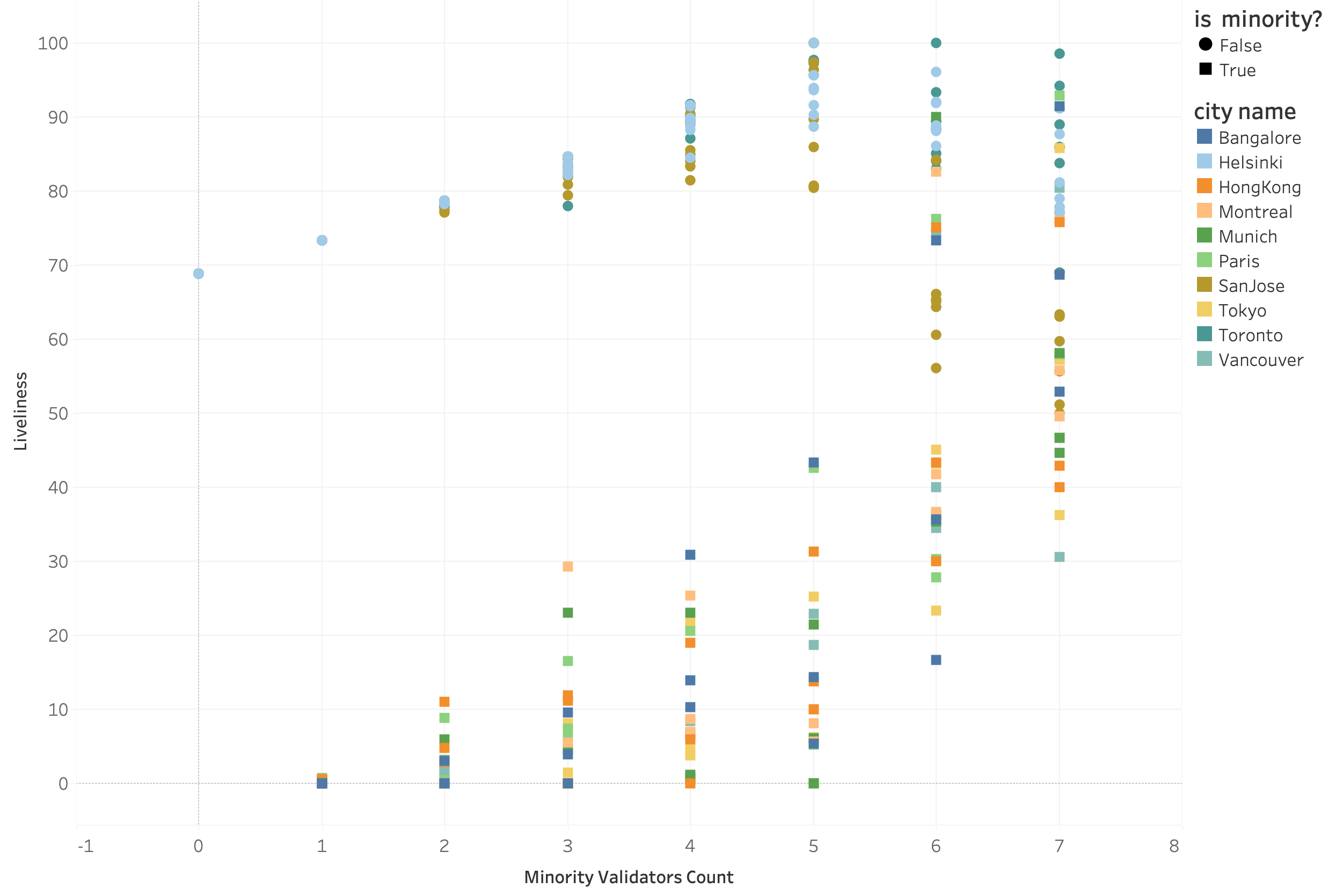}
  \caption{Liveliness of majority and minority validators with varying number of minority validators.}
  \label{fig:theoremPlot}
\end{figure*}
\begin{figure*}[]
  \centering
  \includegraphics[scale =0.26]{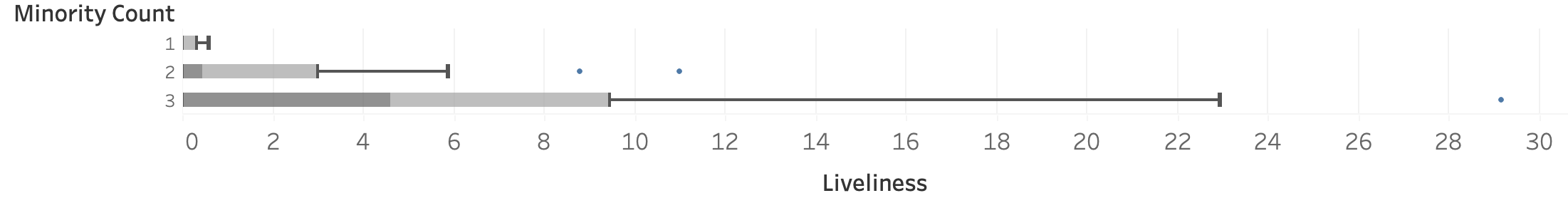}
  \caption{Boxplots of liveliness against minority validators count one, two, three.}
  \label{fig:boxplots}
\end{figure*}
These emulations confirm that the mean liveliness of all minority validators in every run is zero. In contrast, validators from San Jose and Melbourne have a liveliness of 81.79\% and 70.33\%, respectively, as shown in Fig.~\ref{fig:twoMajorities}. Using GeoDec, we are confident that all minority validators are live and not malicious. Additionally, all validators have access to similar computing resources. It is important to remember that our figures are a 2D representation of a spherical globe, and distances between locations may differ from how they appear. Nevertheless, these empirical observations are similar to that observed in 0L. Based on these insights, we now present the following observation.
\begin{theorem} 
Minority validators have the lowest liveliness. 
\end{theorem}
\begin{proof}[Discussion]
Consider a validator set $V$ with minority $\mu = c_m$, and the minority validator set has cardinality less than one-third of total validator set, $\mathbb{C}[\mu] \leq (1/3 * |V|)$. From the viewpoint of the minority, all other validators are so geospatially distant that they appear to be at the same location. Here, the majority $\eta$ refers to every region except the minority, $|V| - \mathbb{C}[\mu] \geq (2/3 * |V|)$.
Every block has a block proposer, and we have two possibilities for every block proposer; they could either be from a majority or a minority. 

Case 1: A majority validator. If we assume all validators are running honestly with no failures, the participation from the minority is zero. Minorities have zero participation because there is always a quorum $\mathbb{Q}$ among majority validators, as their communication delays are minimal. 

Case 2: A minority validator. This block cannot reach a quorum within minority validators and has to be broadcasted to majority validators. Since here we have classical consensus with partial synchrony, there is a high probability that this block round times out and most validators would suspect crash failure of the block proposer. Furthermore, even if these blocks get committed, the liveliness will still be lower because the likelihood of a block proposer being selected from the minority is lower than that of the majority. Hence, the liveliness of minority validators would be the lowest. 
\end{proof}

As a consequence of lower liveliness, the consensus protocol punishes the minority. This punishment is either jailing (removal from the validator set), slashing (penalty in native tokens), or both. We now turn to the experimental evidence in the following subsection. 

\subsection{Empirical evidence: Minority validators have low liveliness}
We employ GeoDec to demonstrate that minority validators have the lowest liveliness empirically. We select a pair of majority and minority validators, initially setting the number of minority validators to zero and gradually increasing their number until they approach half of the validator set's cardinality. Minority validators are selected from the following locations: Bangalore, Hong Kong, Montreal, Munich, Paris, Tokyo, and Vancouver, while most validators are from either Helsinki, San Jose, or Toronto. The validator set is consistently fixed at sixteen, necessitating a quorum of eleven validators to reach a consensus. The results comprise data from 167 runs of five epochs each, grouped by averaging each location's results for each run.

We plot liveliness (X-axis) against the number of minority validators (Y-axis); see Fig.\ref{fig:theoremPlot}. Each run results in two points, a circle for the majority and a square for the minority. The graph unequivocally indicates low liveliness for minority validators. Intriguingly, it also displays a gradual increase in the liveliness of minority validators corresponding to their rising count. We also utilize box plots to display the distribution for the cases of one, two, and three minority counts in Fig.\ref{fig:boxplots}. The median liveliness is 0, 0.46, and 4.61 for one, two, and three minority validators, respectively. None of these meet the liveliness threshold of blockchains, such as 0L ($\pi = 5$). Additionally, we observe a continuous increase in the liveliness of majority validators. This trend is likely because only two-thirds of all co-located, equally capable validators make it to the quorum $\mathbb{Q}$, which explains the peak liveliness of majority validators when the minority count is five $(1/3 * |V|)$. As we approach an equal distribution of minority and majority, their liveliness begins to overlap, further substantiating our hypothesis of low liveliness for minority validators.

\begin{figure}[]
  \centering
  \includegraphics[width=0.46\textwidth]{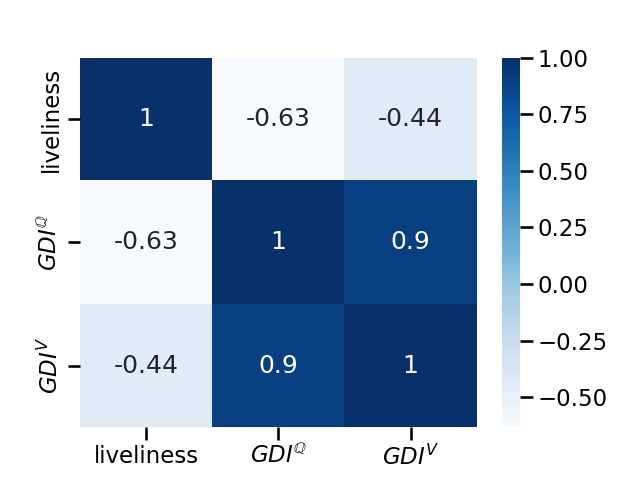}
  \caption{Heatmap of correlation between $GDI^V$, $GDI^\mathbb{Q}$ and liveliness.}
  \label{fig:heatmap}
\end{figure}

We then shift our focus to the GDI to ascertain its utility as a performance indicator for validators. We use the Pearson correlation coefficient to measure the linear correlation between two variables. Fig.~\ref{fig:heatmap} depicts a heatmap that plots the correlation among $GDI_\mathbb{Q}$, $GDI_V$, and liveliness $l$. The negative value indicates a negative correlation between GDI and liveliness, thereby confirming our hypothesis that a higher GDI leads to lower liveliness. Notably, $GDI_\mathbb{Q}$ is a more reliable metric than $GDI_V$. Specifically, the correlation between $GDI_\mathbb{Q}$ and liveliness is -0.631575 with a \textit{pvalue} of $1.6 * 10^{-304}$. These values indicate a strong negative correlation that is statistically significant (\textit{pvalue} $ < 0.05$).

These findings underline that consensus protocols cannot distinguish between crash failures and minority validators. This insight is critical to the design of reward mechanisms in a blockchain. Lower liveliness of minority validators results in penalties and potential removal from the validator set, creating an evolutionary incentive for the geographical centralization of validators. Paradoxically, while blockchains extol the virtues of decentralization, they inadvertently foster geospatial centralization by design.

\subsection{Solution: GDIException Smart Contract}
\label{sec:solution}
The findings of this paper thus far reveal that blockchain consensus protocols, by their very design, inadvertently encourage geospatial centralization. Two critical challenges need to be addressed: first, a revision of the consensus protocols to enhance geospatial diversity during validator selection, and second, the immediate need to refrain from penalizing minority validators who contribute to geospatial diversity in the validator set. This work focuses primarily on the latter issue.

Our approach is to identify the minority validators, reach a consensus on who minority validators are and facilitate their participation, rather than punishing them. We begin by introducing the constructs necessary for our solution. First, we introduce GeoDecLite, a lightweight version of GeoDec, that leverages the IP addresses of the validator set and flags the minority validators. Second, we develop GDIExceptionContract, a smart contract with three core functionalities:
\texttt{nominate()}: allows a validator to nominate itself as a minority validator.
\texttt{vouch()}: enables validators to endorse a minority validator.
\texttt{getMinorities()}: returns a list of nominees who have received a quorum of vouches. 

\begin{figure}[]
  \centering
  \includegraphics[width=0.45\textwidth]{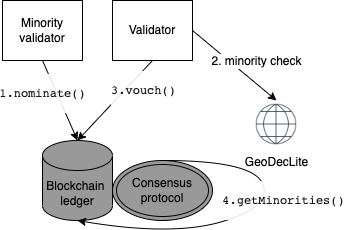}
  \caption{Workflow of the solution.}
  \label{fig:solutionWorkflow}
\end{figure}

We then outline the workflow of our solution in four steps, as illustrated in Fig.~\ref{fig:solutionWorkflow}. 
\begin{enumerate}
  \item During the epoch, the minority validator invokes the \texttt{nominate()} function to propose itself.
  \item All validators, upon detecting this committed event, verify the claim using GeoDecLite.
  \item If the claim is correct, validators utilize the \texttt{vouch()} function to endorse the minority validator.
  \item  At the epoch's conclusion, the consensus protocol invokes \texttt{getMinorities()} before considering penalties or removal from the validator set. Also, the protocol resets the data in the GDIExceptionContract at the onset of a new epoch.
  \end{enumerate}

We have favored this approach over an oracle solution to avoid reliance on external services. Another critical consideration was to prevent the consensus protocol from making external API calls, thus eliminating a potential attack vector. As the validation of minority validators requires a quorum of endorsements, the safety of this approach is on par with the underlying validator set and consensus protocol's safety guarantees. Moreover, deploying this solution on blockchains that support smart contracts is a straightforward process.

\subsection{Evaluation of our solution}
\begin{figure}[]
  \centering
  \includegraphics[width=0.5\textwidth]{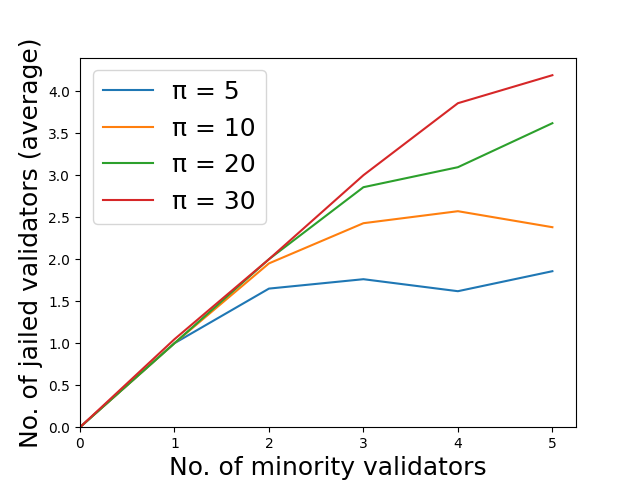}
  \caption{Number of minority validators that are jailed with varying liveliness thresholds $\pi$.}
  \label{fig:livelinessThreshold_solution}
\end{figure}
Now, we turn our attention towards the experimental evaluation of our solution. We delve into the data obtained from previous experiments to comprehend the number of validators who could potentially avoid penalties if our solution were implemented. We experimented with varying liveliness thresholds ($\pi$): 5, 10, 20, and 30. Our analysis indicates that most jailed validators are minorities, all of whom could invoke the \texttt{nominate()} function in GDIException to evade jailing or penalties. Assuming minority validators invoke this straightforward \texttt{nominate()} function call, Fig.~\ref{fig:livelinessThreshold_solution} illustrates the mean number of minority validators that could be spared from jailing. We computed these averages by grouping the total data according to the number of minority validators. Notably, our solution effectively safeguards all minority validators from being jailed, thereby enhancing geospatial distribution.

\begin{figure}[]
  \centering
  \includegraphics[width=0.5\textwidth]{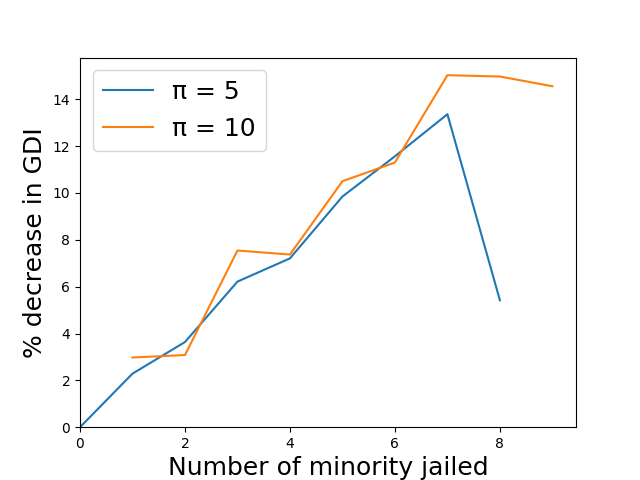}
  \caption{64 node emulations with varying $\pi$.}
  \label{fig:64_node}
\end{figure}

So far, we have examined a scenario with a single minority against a majority. We now expand our evaluations to consider how our solution fares in a scenario with potentially numerous minorities and no single majority. Reflecting observed practices in 0L and Ethereum, we assume a quorum $\mathbb{Q}$ in the validator set $V$ from the US and Europe, with the remainder being geographically distributed outside these regions. Moreover, we increase the size of the validator set to 64 ($|V|$ = 64) for the subsequent emulations. Fig.~\ref{fig:64_node} illustrates the results from 69 runs, plotting the percentage decrease in GDI against the number of jailed minority validators. Our solution aims to address this decline in GDI. 

\section{Related Work}
\label{sec:relatedwork}
To the best of our knowledge, we are among the pioneers in formulating and demonstrating the issue of minority validators being punished, providing empirical evidence from both real-world blockchains and controlled experiments. We acknowledge that the authors of Narwhal and Tusk have noted that minority validators are indistinguishable from faulty ones, presenting this as a limitation of classical BFT consensus protocols~\cite{danezis2022narwhal}. However, their work primarily focuses on presenting consensus mechanism, rather than enhancing geospatial decentralization.

Our work complements the ongoing research on blockchain decentralization~\cite{gochhayat2020measuring,gencer2018decentralization,jia2022measuring,lee2020measurements}. Yet, only a handful of studies concentrate on the geospatial dimension~\cite{zook2022mapping,holicka2021emerging}. In a recent exploration of a stratified approach to blockchain decentralization, geography is considered a critical component, with physical safety and regulatory compliance identified as potential threats posed by geospatial centralization~\cite{karakostas2022sok}. Another research study delving into the impact of geodistribution and mining pools on blockchains demonstrates that a geospatial location significantly affects block reception time~\cite{silva2020impact}. However, these studies primarily focus on PoW blockchains and do not account for factors that could impede the participation of geospatially distant validators in the consensus protocol.

\section{Conclusions}
\label{sec:conclusions}
In this paper, we defined the problem that minority validators are often misclassified as crash failures. We have provided empirical evidence to support this claim in practice and controlled environments. To facilitate these controlled experiments, we developed the GeoDec emulator, a tool whose full potential has yet to be fully tapped in this study. Importantly, we have proposed an easily integrated solution to safeguard minority validators from punishment. We intend to leverage our findings to enhance geospatial decentralization in the validator selection process, aiming to promote it actively, not just preserve it.

\section*{Acknowledgment}
This work has in part been supported by NSERC and ORF.

\bibliographystyle{IEEEtran}
\bibliography{references}

\end{document}